\documentclass{llncs}
\usepackage{amssymb}
\usepackage{url}
\newcommand{\lra}{\longrightarrow}
\newcommand{\s}{\lesssim}

\newcommand{\true}{\mathsf{tt}}
\newcommand{\false}{\mathsf{ff}}
\newcommand{\LogCS}{\mathcal{L}_{CS}}
\newcommand{\Conformance}[2]{#1\lesssim_{CS}#2}

\newcommand{\CoLog}{\mathcal{S}_{CS}}
\newcommand{\OLog}{\preceq}

\newcommand{\LogS}{\mathcal{L}_{S}}
\newcommand{\LogHM}{\mathcal{L}_{HM}}

\newcommand{\LogCC}{\mathcal{L}_{CC}}
\newcommand{\CoCo}[2]{#1\lesssim_{CC}#2}
\newcommand{\CCLog}{\mathcal{S}_{CC}}

\begin{document}

\title{Logics for Contravariant Simulations \thanks{Research supported by the Spanish
projects DESAFIOS10 TIN2009-14599-C03-01, TESIS TIN2009-14321-C02-01 and PROMETIDOS
S2009/TIC-1465.}}

\author{Ignacio F\'abregas \and David de Frutos Escrig
\and Miguel Palomino}
\institute{Departamento de Sistemas Inform\'aticos y
Computaci\'on,  UCM\\
\email{fabregas@fdi.ucm.es \quad \{defrutos, miguelpt\}@sip.ucm.es}}

\maketitle

\begin{abstract}
Covariant-contravariant simulation and conformance simulation are two generalizations of the simple notion of simulation which aim at capturing the fact that it is not always the case that ``the larger the number of behaviors, the better''. Therefore, they can be considered to be more adequate to express the fact that a system is a correct implementation of some specification. We have previously shown that these two more elaborated notions fit well within the categorical framework developed to study the notion of simulation in a generic way. Now we show that their behaviors have also simple and natural logical characterizations, though more elaborated than those for the plain simulation semantics.
\end{abstract}

\section{Introduction and some related work}

Simulations are a very natural way to compare systems modeled by labeled transition systems or other related mechanisms based on describing the behavior of states by means of the actions they can execute \cite{Park81}. They aim at comparing processes based on the simple premise ``you are better if you can do as much as me, and perhaps some additional new things''. This assumes that all the executable actions are controlled by the user (hence, no difference between input and output actions) and does not take into account that  the system will choose in an unpredictable internal way whenever it has several possibilities for the execution of an action; thus, the more possibilities, the less control.

In order to cope with this situation one should consider adequate versions of simulation where the meaning of actions and the idea of preferring processes that are less non-deterministic are taken into account. This leads to two new notions of simulation:
covariant-contravariant simulation and conformance simulation, that we roughly sketched in \cite{DeFrutosEtAl07} and presented in detail in \cite{FabregasFP09}, where we proved that they can be obtained as particular instances of the general notion of categorical simulation developed by Hughes and Jacobs \cite{HughesJacobs04}.

The first new notion is that of covariant-contravariant simulation, where the alphabet of actions $Act$ is partitioned into three disjoint sets $\textit{Act}^l$, $\textit{Act}^r$, and $\textit{Act}^{\mathit{bi}}$. The intention is that simulations will treat the actions in $\textit{Act}^l$ like in the ordinary case, they will interchange the roles of the related processes for
those actions in $\textit{Act}^r$, and they will impose a symmetric condition (like that defining bisimulation) for the actions in
$\textit{Act}^{\mathit{bi}}$. The second notion, conformance simulation, captures the conformance relations \cite{Leduc92} that several authors introduced in order to formalize the notion of possible implementations. 

After showing in \cite{FabregasFP09} that they can be formalized as categorical simulations, in this paper we present their logical characterizations.
We expect that they will contribute to clarify the meaning of the corresponding simulations, shedding light on the properties that can be established when using these two frameworks within a specification procedure.

Certainly, the distinction between input and output actions or similar classifications is not meant to be new at all and, for instance, it was present in modal transition systems as early as the end of the eighties. It also plays a central role in I/O-automata \cite{Lynch88} and, more recently, appears as component of several works on interface automata \cite{AlfaroH01}, where the covariant-contravariant distinction is found when the guarantees of the specification can only be assumed if the conditions of the specification are satisfied.

Concerning conformance simulation, the first related references are also quite old \cite{Leduc92} and correspond to the notion of conformance testing, which is close to failure semantics \cite{VanGlabbeek01}. However, it is a bit surprising that in both cases there is lack of a basic theory where these notions are presented in a simplified scenario, stressing their main characteristics and properties.

Let us conclude this introduction by remarking that there is a large collection of recent papers where notions close to those studied here are either developed or applied. We regret not having the time or space to discuss, or even to cite, many of them and just to give a hint we point out \cite{20years,BenesKLS09}, where several references to other preliminary works in those directions can be found.

\section{Recalling contravariant simulations} 

We consider labeled transition systems (LTS) $(P,A,\rightarrow_P)$, where $\rightarrow_P \subseteq P\times A\times P$, to define the operational semantics of a family of processes $p\in P$. We say that the LTS is \emph{finitary} when for each $p\in P$ and $a\in A$ we have $|\{p'\mid p\stackrel{a}{\lra}p'\}|<\infty$.

We refer to \cite{FabregasFP09} for a more extensive motivation of covariant-contravariant simulations; here we only comment on the case of input/output automata. To define an adequate simulation notion for them we observe that the classic approach to simulations is based on the definition of semantics for reactive systems, where all the actions of the processes correspond to input actions that the user must trigger. Instead, the situation is the opposite whenever we have explicit output actions: it is the system that produces the actions and the user who is forced to accept the produced output. Then, it is natural to conclude that in the simulation framework we have to dualize the simulation condition when considering output actions, and this is exactly what our anti-simulation relations do.

\begin{definition}
Given $P=(P,A,\rightarrow_P)$ and $Q=(Q,A,\rightarrow_Q)$, two labeled transition systems for the alphabet $A$, and $\{A^r,A^l, A^{\mathit{bi}}\}$ a partition of this alphabet, a \textbf{$(A^r,A^l)$-simulation} (or just a covariant-contravariant simulation) between them is a relation $S\subseteq P\times Q$ such that for every $pSq$ we have:
\begin{itemize}
\item for all $a\in A^r\cup A^{\mathit{bi}}$ and all $p\stackrel{a}{\lra}p'$ there exists $q\stackrel{a}{\lra}q'$ with $p'Sq'$.

\item for all $a\in A^l\cup A^{\mathit{bi}}$, and all $q\stackrel{a}{\lra}q'$ there exists $p\stackrel{a}{\lra}p'$ with $p'Sq'$.
\end{itemize}
We will write $\CoCo{p}{q}$ if there exists a covariant-contravariant simulation $S$ such that $pSq$.
\end{definition}

Conformance simulations allow the extension of the set of actions offered by a process, so that in particular  $a\lesssim a+b$, but they also consider that a process can be ``improved'' by reducing the nondeterminism in it, so that $ap+aq\lesssim ap$. In this way we have again a kind of covariant-contravariant simulation, not driven by the alphabet of actions executed by the processes but by their nondeterminism.

\begin{definition}
Given $P=(P,A,\rightarrow_P)$ and $Q=(Q,A,\rightarrow_Q)$ two labeled transition systems for the alphabet $A$, a \textbf{conformance simulation}
between them is a relation $R\subseteq P\times Q$ such that whenever $pRq$, then:
\begin{itemize}
\item For all $a\in A$, if $p\stackrel{a}{\lra}$, then $q\stackrel{a}{\lra}$ (this means, using the usual notation for process
algebras, that $I(p)\subseteq I(q)$).

\item For all $a\in A$ such that $q\stackrel{a} {\lra}q'$ and $p\stackrel{a}{\lra}$, there exists some $p'$ with $p\stackrel{a}{\lra}p'$ and
$p'R q'$.
\end{itemize}
We will write $\Conformance{p}{q}$ if there exists a conformance simulation $R$ such that $pRq$.
\end{definition}

\section{Logical characterizations of the new semantics}\label{sec-Logical}

\subsection{Covariant-contravariant simulations}
The class $\LogS$ characterizing the simulation semantics is defined in \cite{Cirstea06} as that containing $\true$, conjunctions  $\bigwedge_{i\in I}\varphi_i$ (which can be just finite or binary if we only want to characterize finitary process) and the existential operator $\langle a\rangle\varphi$, whose semantics is defined by: $p\models \langle a\rangle \varphi$ if there exists some $p'$ such that $p\stackrel{a}{\lra}p'$ and $p'\models\varphi$.

If we compare it with the Hennessy-Milner logic $\LogHM$ \cite{HennessyMilner85}, it can be noted that the main diference is that negation is not present. Obviously, this must be the case to capture a strict order that is not an equivalence relation, 
such as $\CoCo{}{}$. However,  adding both the constant $\false$ and the disjunction $\bigvee_{i\in I}\varphi_i$ does no harm, 
thus obtaining $\bar{\LogS}$ which also characterizes $\s_S$. 
Indeed, $\false$ is just $\bigvee_{\emptyset}\varphi_i$, while disjunctions can be moved to the top of the expression because  ${\langle a\rangle\bigvee_{i\in I}\varphi_i}\equiv{\bigvee_{i\in I}\langle a\rangle\varphi_i}$, and $p\models \bigvee_{i\in I}\varphi_i$ iff there exists some $i\in I$ such that $p\models \varphi_i$.

The inspiration to obtain the logic characterizing $\CoCo{}{}$ comes from the fact that if we only have contravariant actions, then $\CoCo{}{}$ becomes $\s_S^{-1}$, and therefore by negating all the formulas in $\bar{\LogS}$ we would obtain the desired characterization. In particular, for the modal operator $\langle a\rangle$ we would obtain its dual form $[a]$, whose semantics is defined by: \textit{$p\models [a] \varphi$ if $p'\models\varphi$ for all $p'$ such that $p\stackrel{a}{\lra}p'$}.

Then, in the presence of both covariant and contravariant actions, we need to consider the existential operator $\langle a\rangle$ for $a\in A^r\cup A^\mathit{bi}$ and the universal operator $[a]$ for $a\in A^l\cup A^\mathit{bi}$, thus obtaining the following definition.

\begin{definition}\label{def-sem clas CC}. Given an alphabet $A$, and $\{A^r,A^l, A^{\mathit{bi}}\}$ a partition of this alphabet, the class $\LogCC$ of covariant-contravariant simulation formulas over $A$ is defined recursively by:
\begin{itemize}
\item $\true$ and $\false$ are in $\LogCC$.

\item If $I$ is a set and $\varphi_i\in\LogCC$ for all $i\in I$ then $\bigwedge_{i\in I}\varphi_i\in\LogCC$, $\bigvee_{i\in I}\varphi_i\in\LogCC$.

\item If $\varphi\in\LogCC$ and $a\in A^r\cup A^{\mathit{bi}}$ then $\langle a\rangle\varphi\in\LogCC$.

\item If $\varphi\in\LogCC$ and $a\in A^l\cup A^{\mathit{bi}}$ then $[a]\varphi\in\LogCC$.
\end{itemize}
The satisfaction relation $\models$ is defined recursively by:
\begin{itemize}
\item $p\models\true$.

\item $p\models\bigwedge_{i\in I}\varphi_i$ if $p\models\varphi_i$ for all $i\in I$.

\item $p\models\bigvee_{i\in I}\varphi_i$ if $p\models\varphi_i$ for some $i\in I$.

\item $p\models \langle a\rangle \varphi$ if there exists some $p'$ such that $p\stackrel{a}{\lra}p'$ and $p'\models\varphi$.

\item $p\models [a]\varphi$ if $p'\models\varphi$ for all $p'$ such that $p\stackrel{a}{\lra}p'$.
\end{itemize}
Let $\CCLog(p)$ denote the class of covariant-contravariant simulation formulas satisfied by the process $p$, that is,
$\CCLog(p)=\{\varphi\in\LogCC\mid p\models\varphi\}$. We will write $p\OLog_{CC} q$ if $\CCLog(p)\subseteq\CCLog(q)$.
\end{definition}

The case of input/output transition systems is probably the clearest example where the covariant-contravariant duality must be applied in order to capture the appropriate simulation order. Input actions should have a covariant behavior reflecting the fact that a reactive system is expected to be ``better'' whenever it accepts a maximal set of requests; as a consequence, its logical characterization can only capture liveness properties. Conversely, output actions should be contravariant: whenever we specify a system we expect to control its behavior as much as possible, and outputs are generative, which means not controllable by the user. This contravariant character is captured by the universal operator $[a]$, which is only able to define safety properties. 

Therefore, the logic $\LogCC$ includes formulas that simultaneously capture liveness and safety at a local level, depending on the character of the actions that are used. This is not enough to adequately state all the requirements one could possible need: certainly,  after developing a myriad of different semantics for processes \cite{VanGlabbeek01,DeFrutosEtAl08c}, we would not expect that just by fiddling with one of the simplest, the simulation semantics, we would have the definite answer to treat together covariant and contravariant actions. We are also investigating the covariant-contravariant version of other semantics but, in order to establish which are the basic facts to take into account, it is clear to us that the case of plain simulation is definitely a basic keystone.

\begin{proposition}
$\CoCo{p}{q}\Longleftrightarrow p\OLog_{CC}q$.
\end{proposition}

\begin{proof}
We will first prove the implication from left to right. Assume that we have $pSq$ for some covariant-contravariant simulation
$S$: we must show that for each $\varphi\in\LogCC$, $p\models\varphi$ implies $q\models\varphi$. We proceed by structural induction over $\varphi$.
\begin{itemize}
\item $q\models\true$, trivially.

\item Let $p\models \langle a\rangle\varphi$ with $a\in A^r\cup A^\mathit{bi}$. Then there is $p'$ such that $p\stackrel{a}{\lra}p'$ with
$p'\models\varphi$. Now, since $pRq$ and $a\in A^r\cup A^\mathit{bi}$ there must be a $q'$ such that $q\stackrel{a}{\lra}q'$ with $p'Rq'$ and, by induction hypothesis, $q'\models\varphi$, that is, $q\models \langle a \rangle\varphi$.

\item Let $p\models [a]\varphi$. Then for all $p'$ such that $p\stackrel{a}{\lra}p'$ we have $p'\models\varphi$. Let $q'$ be such that $q\stackrel{a}{\lra}q'$ then, since $pSq$ and $a\in A^l\cup A^\mathit{bi}$, there exists $p'$ such that $p\stackrel{a}{\lra}p'$ and
$p'Sq'$. By induction hypothesis, since $p'\models\varphi$ then $q'\models\varphi$, that is, $q\models [a]\varphi$.

\item Let $p\models\bigwedge_{i\in I}\varphi_i$. Then $p\models\varphi_i$ for all $i\in I$, so by induction hypothesis $q\models\varphi_i$ for
all $i\in I$ and then $q\models\bigwedge_{i\in I}\varphi_i$.

\item $p\models\bigvee_{i\in I}\varphi_i$. It is analogous to the previous case.
\end{itemize}
For the other implication let us assume that $p\OLog_{CC}q$ and show that $\OLog_{CC}$ is a covariant-contravariant simulation. Let $a\in A^r\cup A^{\mathit{bi}}$ and $p\stackrel{a}{\lra}p'$; then there exists $q'$ such that $q\stackrel{a}{\lra}q'$ and $p'\OLog_{CC}q'$. Otherwise, we have that for all $q\stackrel{a}{\lra}q'$,
$p'\not\OLog_{CC}q'$, that is, we have formulas $\varphi_{q'}$ such that $\varphi_{q'}\in\CCLog(p')\setminus\CCLog(q')$. Now, taking
$\phi=\langle a\rangle\bigwedge_{q'}\varphi_{q'}$, we have $p\models\phi$ and, by hypothesis, also $q\models\phi$. That means that there exists some $q'_0$ such that $q\stackrel{a}{\lra}q'_0$ with $q'_0\models\bigwedge_{q'}\varphi_q'$. But this cannot be the case since $q'_0\not\models\varphi_{q'_0}$.

Now let $a\in A^l\cup A^{\mathit{bi}}$ and $q\stackrel{a}{\lra}q'$; similarly we must show that there exists $p'$ such that $p\stackrel{a}{\lra}p'$ and $p'\OLog_{CC}q'$. By way of contradiction, if for all $p\stackrel{a}{\lra}p'$ we have $p'\not\OLog_{CC}q'$, there are formulas $\varphi_{p'}\in\CCLog(p')\setminus\CCLog(q')$. Taking  $\phi=[a]\bigvee_{p'}\varphi_{p'}$ we have $p\models\phi$ and then by hypothesis $q\not\models\phi$, but this cannot be since $q'\not\models\varphi_{p'}$ for all $p'$. \qed
\end{proof}

\subsection{Conformance simulations}
Conformance simulation can be considered to be a variant of the covariant-contravariant framework in which, instead of separating the actions in several classes, we have a mixed uniform behavior for all the actions. This is brought forward by the fact that if a process cannot execute $a$, then $\Conformance{p}{p+aq} $. However, once we have $a\in I(p)$ the contravariant character shows since then $\Conformance{p+aq}{p}$.

This mixed character of all the actions is now captured at the logical level by a new modal operator $a$, whose semantics is defined by:
\textit{$p\models a\varphi$ if $p\stackrel{a}{\lra}$ and $p'\models\varphi$ for all $p\stackrel{a}{\lra}p'$}.
It is quite interesting to observe that we can alternatively define $a$ as ``$\langle a\rangle\wedge[a]$'', since we have: \textit{$p\models a\varphi\;\Longleftrightarrow\; p\models\langle a\rangle\varphi$ and $p\models [a]\varphi$},
which also reveals the mixed intended nature of all the actions in the conformance framework.

\begin{definition}\label{def-sem clas} The class $\LogCS$ of conformance simulation formulas over $A$ is defined recursively by:
\begin{itemize}
\item $\true\in\LogCS$.

\item If $I$ is a set and $\varphi_i\in\LogCS$ for all $i\in I$ then $\bigwedge_{i\in I}\varphi_i,\in\LogCS$, $\bigvee_{i\in I}\varphi_i\in\LogCS$.

\item If $\varphi\in\LogCS$ and $a\in A$ then $a\varphi\in\LogCS$.
\end{itemize}
The corresponding satisfaction relation $\models$ is defined recursively by:
\begin{itemize}
\item $p\models\true$.

\item $p\models\bigwedge_{i\in I}\varphi_i$ if $p\models\varphi_i$ for all $i\in I$.

\item $p\models\bigvee_{i\in I}\varphi_i$ if $p\models\varphi_i$ for some $i\in I$.

\item $p\models a\varphi$ if $p\stackrel{a}{\lra}$ and $p'\models\varphi$ for all $p\stackrel{a}{\lra}p'$.
\end{itemize}
Let $\CoLog(p)$ denote the class of conformance simulation formulas satisfied by the process $p$, that is, $\CoLog(p)=\{\varphi\in\LogCS\mid
p\models\varphi\}$. We will write $p\OLog_{CS} q$ if $\CoLog(p)\subseteq\CoLog(q)$.
\end{definition}

One now expects that the liveness and safety requirements will be captured simultaneously and this is indeed the case since from $p\models a\varphi$ we know both that $p$ is able to execute $a$ and that, after executing it in any possible way,  $\varphi$ will be satisfied. Therefore, conformance simulation proves to be quite a reasonable semantics whenever we do not want to distinguish between reactive and generative actions, as discussed in the previous section.

\begin{proposition}
$\Conformance{p}{q}\Longleftrightarrow p\OLog_{CS}q$.
\end{proposition}

\begin{proof}
We first prove the implication from left to right. Assume that we have $pRq$ for some conformance simulation $R$: we must show that
for each $\varphi\in\LogCS$, $p\models\varphi$ implies $q\models\varphi$. The proof will follow by structural induction over $\varphi$, the case for $\true$ being trivial.
\begin{itemize}
\item Let $p\models a\varphi$. Then, for all $p\stackrel{a}{\lra}p'$ we have $p'\models\varphi$ and there exists at least one such $p'$. Since $pRq$ also $q\stackrel{a}{\lra}$, and it remains to prove that $q'\models \varphi$ for all successors $q\stackrel{a}{\lra}q'$. 
Let $q'_0$ be such that $q\stackrel{a}{\lra}q'_0$. 
Again, since $pRq$ and $p\stackrel{a}{\lra}$, for each $q\stackrel{a}{\lra}q'$ there exists some $p\stackrel{a}{\lra}p'$ such that $p'Rq'$. So, for $q'_0$ there exists
$p'_0$ such that $p'_0Rq'_0$ and, since $p'_0\models\varphi$, by induction hypothesis also $q'_0\models\varphi$. Thus
$q\models a\varphi$.

\item Let $p\models\bigwedge_{i\in I}\varphi_i$. Then $p\models\varphi_i$ for all $i\in I$, so by induction hypothesis 
$q\models\varphi_i$ for all $i\in I$ and then $q\models\bigwedge_{i\in I}\varphi_i$.

\item $p\models\bigvee_{i\in I}\varphi_i$. It is analogous to the previous case.
\end{itemize}

For the other implication, let us assume that $p\OLog_{CS}q$: we show that $\OLog_{CS}$ is a conformance simulation. First, if
$p\stackrel{a}{\lra}$ then, since $\CoLog(p)\subseteq\CoLog(q)$ and $p\models a\true$, also $q\models a\true$ and hence $q\stackrel{a}{\lra}$. Now, let $q\stackrel{a}{\lra}q'$ and $p\stackrel{a}{\lra}$. Let us see that there exists some $p'$ such that $p\stackrel{a}{\lra}p'$ and $p'\OLog_{CS}q'$. By way of contradiction, if $p'\not\OLog_{CS}q'$ for all such $p'$, then for each $p'$ there is a formula  $\varphi_{p'}\in\CoLog(p')\setminus\CoLog(q')$. Let $\phi=a\bigvee_{p'}\varphi_{p'}$. It is easy to see that $p\models\phi$: indeed, for each $p'$ such that $p\stackrel{a}{\lra}p'$, $p'\models\varphi_{p'}$. Since  $p\OLog_{CS}q$, it must also be the case that $q\models\phi$, that is, for each $q''$ such that $q\stackrel{a}{\lra}q''$, $q''\models\bigvee_{p'}\varphi_{p'}$; but $q\stackrel{a}{\lra}q'$ and $q'\not\models\varphi_{p'}$ for any $p'$, contradicting the fact that $q\models\phi$. \qed
\end{proof}

\section{Some examples and a short discussion}

We will start by illustrating the behavior of covariant-contravariant simulations in the case in which we distinguish between input (reactive) and output (generative) actions. Consider the following expending machines:

\begin{tabular}{l@{\hskip 1cm}l@{\hskip .2cm}l}
$\mathsf{onecoke}$ & : & $\mathsf{coin}\rightarrow \mathsf{coke}\rightarrow 0$\\
$\mathsf{cokeorlemonade}$ & : & $\mathsf{coin}\rightarrow \big( (\mathsf{coke}\rightarrow 0)+(\mathsf{lemonade}\rightarrow 0)\big)$
\end{tabular}

\noindent The classical approach would consider $\mathsf{onecoke}\s_S\mathsf{cokeorlemonade}$. However, if the drinks are provided by the machine in an autonomous way then they should be formalized as outputs, which leads us to
\[
\CoCo{\mathsf{cokeorlemonade}}{\mathsf{onecoke}}.
\]
This  is justified by the fact that choices between generative actions become internal and therefore generate (undesired) non-deterministic behavior.

At the logical level the difference between the two processes above can be brought forward by means of the formula $\langle\mathsf{coin}\rangle\;[\mathsf{lemonade}]\;\false$, which $\mathsf{onecoke}$ satisfies but {\sf cokeorlemonade} does not. It could be thought that the process $\mathsf{cokeorlemonade}$ is being punished for offering lemonade besides coke, but this would be an incorrect interpretation because it follows the classical reactive approach where simultaneous offers mean ``the user makes his choice''; instead, when outputs are generative it is the machine that chooses. As a consequence, from $\mathsf{cokeorlemonade}\not\models\langle\mathsf{coin}\rangle[\mathsf{lemonade}]\;\false$ we implicitly infer that it could be the case that after inserting a coin we did not get our favorite drink (Coke).

Let us now show the differences between covariant-contravariant and conformance simulations. First, at the formal level, the fact that the modal operator $a$ can be defined as ``$\langle a\rangle\wedge[a]$'' does not mean that these two basic modal operators can appear separately in a formula characterizing $\Conformance{}{}$. Obviously this cannot be the case since separated $\langle a\rangle$ operators characterize plain simulation, and for the process $\mathsf{choice\_coke\_lemonade}$: $(\mathsf{coin}\rightarrow \mathsf{coke}\rightarrow 0)+(\mathsf{coin}\rightarrow \mathsf{lemonade}\rightarrow 0)$ we have
\[\mathsf{choice\_coke\_lemonade}\models \langle\mathsf{coin}\rangle\langle\mathsf{lemonade}\rangle\true\qquad
\mathsf{onecoke}\not\models \langle\mathsf{coin}\rangle\langle\mathsf{lemonade}\rangle\true
\] but $\Conformance{\mathsf{choice\_coke\_lemonade}}{\mathsf{onecoke}}$. 

Now, if we consider de universal operator $[a]$, its weakness when used alone arises when it is trivially satisfied. For instance, we have $0\models[\mathsf{coin}]\;\false$ but $\mathsf{onecoke}\not\models[\mathsf{coin}]\;\false$ and $\Conformance{0}{\mathsf{onecoke}}$.

One could infer that conformance simulation is the definitive solution to capture all the natural requirements in an specification. Certainly, it combines covariant and contravariant aspects in a very balanced way, but the fact that it treats all the actions uniformly makes it impossible to capture the difference between input and output actions. In particular: $\Conformance{\mathsf{onecoke}}{\mathsf{cokeorlemonade}}$ but we have already discussed that when outputs are generative, choices always generate non-deterministic behaviors that $\Conformance{}{}$ is not punishing at all.

On the other hand, choices between equal actions are also considered ``harmful'' by the conformance semantics so that if $\Conformance{p}{q}$ then ${ap}=_{CS}{ap+aq}$. This is sometimes a too pessimistic approach, which we can illustrate by the following {\sf slot\_machine} specification: 
\[\mathsf{slot\_machine}:(\mathsf{coin}\rightarrow\mathsf{souvenir}\rightarrow 0)+ (\mathsf{coin}\rightarrow((\mathsf{million\$}\rightarrow 0)+ (\mathsf{souvenir}\rightarrow 0)))
\]

\noindent which becomes conformance simulation equivalent to the {\sf pluff\_machine}
\[\mathsf{pluff\_machine}:\mathsf{coin}\rightarrow\mathsf{souvenir}\rightarrow 0
\]

\noindent In this case the possible return of the big pot is not taken into account at all. Obviously, the solution comes from choosing in each case the adequate semantics to capture accurately the desired behaviors. The bad news is that we need to study many diferent semantics; the good news for us is\dots the same!, since we are already working on them



\begin{thebibliography}{10}

\bibitem{20years}
A.~Antonik, M.~Huth, K.~Larsen, U.~Nyman, and A.~Wasowski.
\newblock {20 Years of Mixed and Modal Specifications}.
\newblock {\em Bulletin of the European Association for Theor. Comput.
  Sci.}, May 2008.

\bibitem{BenesKLS09}
N.~Benes, J.~Kret\'{\i}nsk{\'y}, K.~G. Larsen, and J.~Srba.
\newblock On determinism in modal transition systems.
\newblock {\em Theor. Comput. Sci.}, 410(41):4026--4043, 2009.


\bibitem{Cirstea06}
C.~C\^{\i}rstea.
\newblock A modular approach to defining and characterising notions of
  simulation.
\newblock {\em Inf. Comput.}, 204(4):469--502, 2006.

\bibitem{AlfaroH01}
L.~de~Alfaro and T.~A. Henzinger.
\newblock Interface automata.
\newblock In {\em ESEC / SIGSOFT FSE}, pages 109--120, 2001.

\bibitem{DeFrutosEtAl08c}
D.~de~Frutos~Escrig, C.~Gregorio-Rodr{\'\i}guez, and M.~Palomino.
\newblock On the unification of semantics for processes: observational
  semantics.
\newblock In {\em SOFSEM 09, Proceedings}, {\em LNCS 5404}, pages 279--290. Springer, 2009.

\bibitem{DeFrutosEtAl07}
D.~de~Frutos-Escrig, F.~Rosa~Velardo, and C.~Gregorio-Rodr\'{\i}guez.
\newblock New bisimulation semantics for distributed systems.
\newblock In {\em FORTE 2007, Proceedings}, {\em LNCS 4547}, pages 143--159.
  Springer, 2007.

\bibitem{FabregasFP09}
I.~F{\'a}bregas, D.~de~Frutos-Escrig, and M.~Palomino.
\newblock Non-strongly stable orders also define interesting simulation
  relations.
\newblock In {\em CALCO 09, Proceedings}, {\em LNCS 5728}, pages 221--235. Springer,
  2009.

\bibitem{HennessyMilner85}
M.~Hennessy and R.~Milner.
\newblock Algebraic laws for nondeterminism and concurrency.
\newblock {\em J. ACM}, 32(1):137--161, 1985.

\bibitem{HughesJacobs04}
J.~Hughes and B.~Jacobs.
\newblock Simulations in coalgebra.
\newblock {\em Theor. Comput. Sci.}, 327(1-2):71--108, 2004.



\bibitem{Leduc92}
G.~Leduc.
\newblock A framework based on implementation relations for implementing
  {LOTOS} specifications.
\newblock {\em Computer Networks and ISDN Systems}, 25(1):23--41, 1992.

\bibitem{Lynch88}
N.~Lynch.
\newblock I/O automata: A model for discrete event systems.
\newblock In {\em 22nd Annual Conference on Information Sciences and Systems} pages 29--38, 1988.

\bibitem{Park81}
D.~Park.
\newblock Concurrency and automata on infinite sequences.
\newblock In {\em Theor. Comput. Sci. 5th
  GI-Conference, Proceedings}, {\em LNCS 104}, pages 167--183. Springer,
  1981.

 

\bibitem{VanGlabbeek01}
R.~J. van Glabbeek.
\newblock The linear time-branching time spectrum {I}: The semantics of
  concrete, sequential processes.
\newblock In {\em Handbook
  of process algebra}, pages 3--99. North-Holland, 2001.

\end{thebibliography}

\end{document}